\documentclass[conference]{IEEEtran}
\IEEEoverridecommandlockouts
\usepackage[noadjust]{cite}

\usepackage{amsmath,amssymb,amsfonts}
\usepackage{algorithmic}
\usepackage{graphicx}
\usepackage{textcomp}
\usepackage{xcolor}
\usepackage{amsthm}
\usepackage{bbm}

\newtheorem{theorem}{Theorem}[section]

\newtheorem{lemma}[theorem]{Lemma}

\begin{document}

\title{Sequential Change Detection through Empirical Distribution and Universal Codes}

\author{
Vikrant Malik \\
Deptt. of Electrical Engineering \\
Indian Institute of Technology Kanpur, India \\
vikrant@iitk.ac.in
\and
R. K. Bansal \\
Deptt. of Electrical Engineering \\
Indian Institute of Technology Kanpur, India \\
rkb@iitk.ac.in
}

\maketitle

\begin{abstract}
Universal compression algorithms have been studied in the past for sequential change detection, where they have been used to estimate the post-change distribution in the modified version of the Cumulative Sum (CUSUM) Test. In this paper, we introduce a modified CUSUM test where the pre-change distribution is also unknown and an empirical version of the pre-change distribution is used to implement the algorithm. We present a study of various characteristics of this modified CUSUM Test and then prove its asymptotic optimality.
\end{abstract}

\section{Introduction}
Change detection has widespread applications especially in the medical and engineering fields and it has been widely studied. In this aspect, Page's Cumulative Sum (CUSUM) Test has been a well known tool to approach to the problem of change detection in non-bayesian framework. The performance of Page's CUSUM test has also been extensively studied throughout the years and it's asymptotic properties have been established in \cite{lai}, \cite{lorden} and \cite{bansal86}.

However, Page's CUSUM test requires the knowledge of both pre-change and post-change distributions. Having knowledge of both of these distributions might not always be the case in real-world scenario. For example, in the problem of fault detection, the post-change distribution is likely to be unknown. To address this challenge, ideas from the field of information theory were used by Jacob and Bansal in \cite{jacob2008}. The authors introduced a modified CUSUM test where they made use of a universal code in place of the log-likelihood of the post-change distribution. The authors also proved the optimality of this test in the asymptotic regime and showed that this new test was asymptotically equivalent to the original Page's CUSUM test.

This work was studied further in \cite{srivastava2018sequential} where the authors studied the asymptotic performance of a modified version of the JB-Page Test under Lai's criterion for probability of false alarm under a window. In \cite{chittam}, the authors used the concept of sequential change detection through universal codes for universal compression of a piece-wise stationary source. Further, in \cite{verma2019sequential}, the authors extended the work of \cite{jacob2008} and studied the asymptotic performance of the proposed test for markov sources. 

In this paper, we introduce a further modification of this test which addresses the unknown character of the pre-change distribution without affecting the estimate of the post-change distribution. We then prove the asymptotic optimality of this test.

\section{Preliminaries}
A random process given by $X=\left\{X_{n}\right\}_{1}^{\infty}$ takes values in a finite set called the source alphabet $\mathcal{X}$. We denote the sequence of random variables $X_{i}, X_{i+1}, \ldots, X_{j}$ by $X_{i}^{j}$ whereas a sequence $x_{i}, x_{i+1}, \ldots, x_{j}$ of source symbols is denoted by $x_{i}^{j}$. We also denote the probability distribution and the $n$-dimensional marginal distribution of the random process $X=\left\{X_{n}\right\}_{n=1}^{\infty}$ by $\mu$. Thus, $\operatorname{Pr}\left\{X_{1}^{n}=x_{1}^{n}\right\}=\mu\left(x_{1}^{n}\right)$

A fixed to variable length code ( FV code ) is a one to one mapping $\varphi_{n}: \mathcal{X}^{n} \rightarrow\{0,1\}^{*}$ for a given $n$. Here $\{0,1\}^{*}$ is the set of all binary strings of finite length. For such a code, the length function $L\left(x_{1}^{n}\right)=\left|\varphi_{n}\left(x_{1}^{n}\right)\right|$ denotes the length of the codeword generated by this code for a given sequence $x_{1}^{n}$.

A prefix-free set is a collection of strings such that no string in the set is a prefix of another string in the same set. If $\varphi_{n}\left(\mathcal{X}^{n}\right)$ is a prefix-free set (or simply, a prefix set), then by Kraft's inequality, we have,

\begin{equation}
\sum_{x_{1}^{n} \in \mathcal{X}^{n}} 2^{-L\left(x_{1}^{n}\right)} \leq 1.
\label{eq:kraft}
\end{equation}

The entropy rate of a random process $X$ with distribution $\mu$ is given by 
\begin{equation}
H(\mu)=-\lim _{n \rightarrow \infty} \frac{1}{n} \sum_{x_{1}^{n} \in \mathcal{X}^{n}} \mu\left(x_{1}^{n}\right) \log \mu\left(x_{1}^{n}\right).
\end{equation}

Moreover, the Kullback-Leiber Divergence between two stationary processes $\mu$ and $\nu$ is given by,

\begin{equation}
D(\mu \| \nu)=\lim _{n \rightarrow \infty} \frac{1}{n} \sum_{x_{1}^{n} \in \mathcal{X}^{n}} \mu\left(x_{1}^{n}\right) \log \frac{\mu\left(x_{1}^{n}\right)}{\nu\left(x_{1}^{n}\right)}.
\end{equation}

In our analysis, we assume memory-less sources. In that case, we have,

\begin{equation}
H(\mu)= -\sum_{x \in \mathcal{X}} \mu\left(x\right) \log \mu\left(x\right)
\end{equation}

and,
\begin{equation}
D(\mu \| \nu)= \sum_{x \in \mathcal{X}} \mu\left(x\right) \log \frac{\mu\left(x\right)}{\nu\left(x\right)}.
\end{equation}

For a stationary process, there is a well known lower bound for the compression ratio of any lossless compression algorithm. This bound is nothing but the entropy rate of the source \cite{cover}. Optimal Codes achieve this lower bound. Universal Codes constitute a prefix class which is asymptotically optimal. For such codes, and for a class of memory-less sources $\mathcal{M}$, a FV code with length function $L\left(x_{1}^{n}\right)$ is said to be (strongly) universal if,
\begin{equation}
R_{L}^{n}=\sup _{\mu \in \mathcal{M}} \max _{x_{1}^{n} \in \mathcal{X}^{n}}\left(L\left(x_{1}^{n}\right)+\log \mu\left(x_{1}^{n}\right)\right)=o(n).
\end{equation}

\subsection{Previous Work}
In our test setting, the pre-change distribution is given by $\mu_0$ and the post change distribution is given by $\mu_1$ and the samples generated by both are i.i.d. For the case when both the pre change and post change distributions are known, Page proposed a test \cite{page1954} which is given as follows,

\textbf{Page's CUSUM Test}:
Test 1: Starting with $n = 1$, for samples $x_1, x_2, \ldots, x_n$, and $\gamma > 1$, if
\begin{equation}
\max _{k \leq n}\left(\log \mu_{1}\left(x_{k}^{n}\right)-\log \mu_{0}\left(x_{k}^{n}\right)\right) \geq \log \gamma,
\end{equation}
stop the test and decide that a change has occurred. Else, sample for $x_{n+1}$ and continue the test.

\textbf{Jacob Bansal Page CUSUM Test}:
Since Page's CUSUM test requires the knowledge of both pre-change and post change distributions, a modification of this test was proposed in \cite{jacob2008} for cases where the post-change distribution $\mu_1$ is unknown. The modified CUSUM Test, where a universal code is used to estimate the post-change distribution is given as follows,

Test 2: Starting with $n = 1$, for samples $x_1, x_2, \ldots, x_n$, $\gamma > 1$ and $\lambda < D\left(\mu_{1} \| \mu_{0}\right)$, if

\begin{equation}
\max _{k \leq n}\left(-L\left(x_{k}^{n}\right)-\log \mu_{0}\left(x_{k}^{n}\right)  - n \lambda \right) \geq \log \gamma,
\end{equation}

stop the test and decide that a change has occurred. Else, sample for $x_{n+1}$ and continue the test.

An auxiliary stopping time for this test is defined as,
\begin{equation*}
N_0(\alpha)=\inf \left\{n:-L\left(x_{1}^{n}\right)-\log \mu_{0}\left(x_{1}^{n}\right)-n \lambda \geq-\log \alpha\right\}
\end{equation*}

An error occurs when Test 2 stops even though no change in distribution has occurred. Using the stopping time $N_0(\cdot)$, we can calculate the probability of error for various tests. For the JB-Page Test in \cite{jacob2008}, the following results were shown.

\begin{lemma}
The probability of error, for Test 2 is given as,
\begin{equation}
\mu_{0}(N_0(\alpha)<\infty) \leq \alpha /\left(2^{\lambda}-1\right).
\end{equation}
\label{lemma:pe}
\end{lemma}
Also, the stopping time for Test 2 is given as $\left(\text{here, we use } \alpha = \frac{1}{\gamma}\right)$,
\begin{equation*}
M_0(\gamma)=\inf \{n: \max_{k \leq n}(-L\left(X_{k}^{n}\right)-\log \mu_0\left(X_{k}^{n}\right)-n \lambda) \geq \log \gamma\}.
\end{equation*}

\begin{lemma}
For Test 2 with stopping time $M_0(\gamma)$, we have,
\begin{equation}
\mathbb{E}_{0}(M_0(\gamma)) \geq \gamma\left(2^{\lambda}-1\right),
\end{equation}

and as $\gamma \rightarrow \infty$, we have,
\begin{equation}
\bar{\mathbb{E}}_{1}(M_0(\gamma)) \sim|\log \gamma| /\left(D\left(\mu_{1} \| \mu_{0}\right)-\lambda\right).
\end{equation}
\end{lemma}

\section{Modified CUSUM Test}

\subsection{Estimation of Pre-Change Distribution}
For a given stream of source symbols, we are given that there is no change upto the $n_0$ symbols. Our tests now consists of two stages. First, we estimate the pre-change distribution using these $n_0$ symbols. Although the symbols are from a single stream, we denote the first $n_0$ symbols by $x_1, x_2 \cdots x_{n_0}$ and the subsequent symbols by $x_{n_0 + 1}, x_{n_0 + 2}, \cdots x_{n_0 + n}$. Using these samples, an estimate for $\mu_0$, given by $\hat{\mu}_0$ can be calculated as follows,

\begin{equation}
    \hat{\mu}_0(a) = \frac{\text{N}(a \mid x_1^{n_0})}{n_0} \; \forall a \in \mathcal{X}.
\end{equation}

Here, $\text{N}(a \mid x_1^{n_0})$ represents the number of times the symbol $a$ appears in the sequence $x_1^{n_0}$. Therefore, for any symbol $a \in \mathcal{X}$, the estimate for $\mu_0(a)$ is given as the fraction of times $a$ appears in the sequence $x_1^{n_0}$.

Here, we also assume that $n_0$ is large enough such that each symbol $a \in \mathcal{X}$ which has non-zero probability under $\mu_0$ i.e., $\mu_0(a) \neq 0$ has appeared at least once in $x_1^{n_0}$. If this is not the case, then the empirical distribution $\hat{\mu}_0$ is not an accurate representation of $\mu_0$.

Since the process is i.i.d., by Weak Law of Large Numbers, we can say that there exists an $n_0$ such that
\begin{equation}
    \mu_0(\left\{x_1^{n_0} : |\hat{\mu}_0(a) - \mu_0(a)| > \delta \right\}) \leq \epsilon_0 \text{ for one or more } \text{a} \in \mathcal{X}.
    \label{eq:slln}
\end{equation}

\subsection{Test 3}
Using the estimate $\hat{\mu}_0$, we now define the modified CUSUM Test as, follows,

Test 3: Starting with $n = 1$, for samples $x_{n_0+1}, x_{n_0+2},$ $\ldots, x_{n_0+n}$ and $\gamma > 1$ if
\begin{equation}
\max _{1 \leq k \leq n}\left(-L\left(x_{n_0+k}^{n_0+n}\right)-\log \hat{\mu}_0 \left(x_{n_0+k}^{n_0+n}\right) - n \lambda \right) \geq \log \gamma,
\end{equation}

stop the test and decide that a change has occurred. Else, sample for $x_{n_0+n+1}$ and continue the test.
Note that for this test, the upward drift rate is $D\left(\mu_{1} \| \hat{\mu}_0\right) - \lambda$ and the downward drift rate is $\lambda - D\left(\mu_{0} \| \hat{\mu}_0\right)$. This means that for the test to work, we need $D\left(\mu_{0} \| \hat{\mu}_0\right) < \lambda < D\left(\mu_{1} \| \hat{\mu}_0\right)$. This also implies that $D\left(\mu_{0} \| \hat{\mu}_0\right) < D\left(\mu_{1} \| \hat{\mu}_0\right)$.

An auxiliary stopping time for this test is defined as,
\begin{align}
N(\alpha)=\inf &\left\{n:-L\left(x_{n_0+1}^{n_0 + n}\right)-\log \hat{\mu}_{0}\left(x_{n_0+1}^{n_0+n}\right) \right. \nonumber \\
&\left. \text{     } -n \lambda \geq-\log \alpha\right\}.
\label{eq:nalpha}
\end{align}
\subsection{Behaviour Under Pre-Change Distribution}
To obtain various performance measures, we will first study the expected behaviour of the stopping time defined earlier under pre-change distribution $\mu_0$ and post-change distribution $\mu_1$. For this, we define the following notations,
\begin{gather*}
f_n(x_{n_0 + 1}^{n_0 + n}) \triangleq \frac{\log\left( \frac{\mu_0(x_{n_0 + 1}^{n_0 + n})}{\hat{\mu}_0(x_{n_0 + 1}^{n_0 + n})} \right)}{n} \\
l_n(x_{1}^{n_0 + n}) \triangleq-L\left(x_{n_0 + 1}^{n_0 + n}\right)-\log \hat{\mu}_0\left(x_{n_0 + 1}^{n_0 + n}\right)-n \lambda \\
B_{n} \triangleq\left\{x_{1}^{n_0 + n}: l_k(x_{1}^{n_0 + k})<-\log \alpha, 1\leq k<n, \right. \\
\left. l_n(x_{1}^{n_0 + n}) \geq-\log \alpha\right\} \\
B \triangleq \cup_{n = 1}^{\infty} B_n \\
E_n^{\delta} \triangleq\left\{x_{1}^{n_0 + n} : |\hat{\mu}_0(a) - \mu_0(a)| < \delta \text{  } \forall \text{a} \in \mathcal{X} \right. \\
\left. \text{ where } \hat{\mu}_0 \text{ is based on samples } x_{1}^{n_0} \right\} \\
E^{\delta} \triangleq \cup_{n = 1}^{\infty} E_n^{\delta}
\end{gather*}
Note that for the set $B_n$, the symbols $x_1^{n_0}$ can be arbitrary since they don't play a role in the constraint equation. Similarly, for $E_n^{\delta}$, the symbols $x_{n_0 + 1}^{n_0 + n}$ can be arbitrary since the empirical distribution is estimated using the first $n_0$ symbols.

Having defined the above notations, we give the following lemma.
\begin{lemma}
If $x_{1}^{n_0+n} \in  E_n^{\delta}$, then $f_{n}(x_{n_0 + 1}^{n_0 + n})$ is bounded above as (for all n)
\begin{equation}
    f_{n} \leq \log (\beta),
    \label{eq:boundmain}
\end{equation}
where $\beta=\frac{\mu_{0}(a)}{\mu_{0}(a)-\delta}$ and $a$ is such that $a$ is the minimum probability symbol. Also, as $\delta \rightarrow 0$ and $n_0 \rightarrow \infty$, we have, $\beta \rightarrow 1$.
\label{lemma:bound}
\end{lemma}
\begin{proof}
Note that since the symbols $x_i's$ are independent and identically distributed, we can write $f_n(x_{n_0 + 1}^{n_0 + n})$ as
\begin{equation}
    f_n(x_{n_0 + 1}^{n_0 + n}) = \sum_{a \in \mathcal{X}} \frac{\text{N}(a \mid x_{n_0 + 1}^{n_0 + n})}{n} \cdot \log\left(\frac{\mu_0(a)}{\hat{\mu}_0(a)}\right).
\end{equation}
Now, since $x_1^{n_0 + n} \in E_n^{\delta}$, by the definition of $E_n^{\delta}$ we have, for any symbol $a \in \mathcal{X}$ sitting at any position in $x_i$ for $i = n_0 + 1, n_0 + 2, \ldots, n_0 + n$,
\begin{equation}
    \frac{\mu_0(a)}{\hat{\mu}_0(a)} < \frac{\mu_0(a)}{\mu_0(a) - \delta} \cdot
\end{equation}
The term on the right hand side is a decreasing function of $\mu_0(a)$. Therefore, if $a'$ is the minimum probability symbol, we can bound the above expression for $f_n(x_{1}^{n_0 + n})$ by (for $\delta < \mu_0(a)$),
\begin{align*}
    f_n(x_{n_0 + 1}^{n_0 + n}) &\leq \sum_{a \in \mathcal{X}} \frac{\text{N}(a \mid x_{n_0 + 1}^{n_0 + n})}{n} \cdot \log \left(\frac{\mu_0(a')}{\mu_0(a') - \delta}\right) \\
    &\leq \log \left(\frac{\mu_0(a')}{\mu_0(a') - \delta}\right) = \log(\beta).
    \label{eq:boundE}
\end{align*}
\end{proof}
We now derive the general expression for the probability of error for the proposed test. Formally, this is stated as follows.
\begin{lemma}
The probability of error for the proposed test is bounded and is given by
\begin{equation}
    \mu_0(N(\alpha) < \infty) \leq \frac{\alpha}{2^{\lambda - \log(\beta)} - 1} + \epsilon_0,
    \label{eq:pe}
\end{equation}
where $\epsilon_0\rightarrow 0$ and $\beta \rightarrow 1$ as $n_0 \rightarrow \infty$.
\end{lemma}

\begin{proof}
The probability of error is defined as the probability that the test terminates and we detect a change in distribution even though there was no actual change in the distribution of the source. Also, note that using the definition of $f_n(x_{n_0+ 1}^{n_0 + n})$, we can write,
\begin{equation}
    \mu_0(x_{n_0 + 1}^{n_0 + n}) = \hat{\mu}_0(x_{n_0 + 1}^{n_0 + n}) \cdot 2^{n f_n(x_{n_0+ 1}^{n_0 + n})}.
\end{equation}

Using this observation, the probability of error can be given as,
\begin{align}
\mu_{0}(N(\alpha)<\infty) \nonumber &= \mu_0\left(B\right) = \mu_0\left(B \cap E^{\delta}\right) + \mu_0\left(B \cap (E^{\delta})^c\right) \nonumber \\
&=\sum_{n=1}^{\infty} \mu_{0}\left(B_{n} \cap E_n^{\delta} \right) + \mu_0\left(B \cap (E^{\delta})^c\right) \nonumber \\
&=\sum_{n=1}^{\infty} \left( \sum_{x_{1}^{n_0 + n} \in B_{n} \cap E_n^{\delta}} \mu_{0}\left(x_{1}^{n_0 + n}\right) \right) \nonumber  \\
&+ \mu_0\left(B \cap (E^{\delta})^c\right).
\label{eq:totaleq}
\end{align}
Here, we have split the total probability into two cases, since the corresponding events are disjoint. Also, note that $B_n \cap E^{\delta} = B_n \cap E_n^{\delta}$, since $B_n$ consists of only $n$-long sequences.

Now, we will calculate bounds on each of the terms of Equation \eqref{eq:totaleq} individually. Starting with the first term, we get,
\begin{align}
&\sum_{ x_{1}^{n_0 + n} \in B_{n} \cap E^{\delta}} \mu_{0}\left(x_{1}^{n_0 + n}\right) \leq \sum_{ x_{1}^{n_0 + n} \in B_{n} \cap E^{\delta}} \mu_{0}\left(x_{n_0 + 1}^{n_0 + n}\right) \\
&\leq\sum_{ x_{1}^{n_0 + n} \in B_{n} \cap E_n^{\delta}}  \hat{\mu}_0(x_{n_0 + 1}^{n_0 + n}) \cdot 2^{n f_{n}\left(x_{n_0 + 1}^{n_0 + n}\right)} \\
&\leq \sum_{x_{1}^{n_0 + n} \in B_{n} \cap E_n^{\delta}}  \hat{\mu}_0(x_{n_0 + 1}^{n_0 + n}) \cdot 2^{n \log (\beta)} \\
&\leq \sum_{x_{1}^{n_0 + n} \in B_{n} \cap E_n^{\delta}}  2^{\log \alpha-L\left(x_{n_0 + 1}^{n_0 + n}\right)-n \left(\lambda - \log (\beta)\right)} \\
&= \alpha \cdot 2^{-n \left(\lambda - \log (\beta)\right)} \left( \sum_{x_{1}^{n_0 + n} \in B_{n} \cap E_n^{\delta}}  2^{-L\left(x_{n_0 + 1}^{n_0 + n}\right)} \right) \\
&\leq \alpha \cdot 2^{-n \left(\lambda - \log (\beta)\right)}.
\end{align}
Here, the first inequality comes from the fact that the probability is bounded as $(1 - \epsilon_0) \leq \mu_0(x_1^{n_0}) \leq 1$. This still gives a tight bound on the probability of error because of Equation \eqref{eq:slln}. The second inequality is due to Lemma \ref{lemma:bound} since $x_1^{n_0 + n} \in E_n^{\delta}$. The third inequality results from the definition of $B_n$ and the last inequality is due to Kraft's Inequality (See Equation \ref{eq:kraft}).

Also, since $\mathbb{P}(A \cap B) \leq \mathbb{P}(B)$, using Equation \eqref{eq:slln}, we have,
\begin{align}
\mu_{0}\left(B \cap (E^{\delta})^c\right) &\leq \mu_{0}\left((E^{\delta})^c\right) \\
&\leq \epsilon_0.
\end{align}

Combining all these equations, we get,
\begin{align}
    \mu_{0}(N(\alpha)<\infty) &\leq \alpha \sum_{n=1}^{\infty} 2^{-n \left(\lambda - \log(\beta)\right)} + \epsilon_0 \\
    &= \frac{\alpha}{2^{\lambda -\log(\beta)} - 1} +  \epsilon_0\cdot
\end{align}
On comparing this expression with the probability of error for the JB-Page's test (Test (2)), we can see that the probability of error increases in our case and it is now a function of $\log (\beta)$. Moreover, for the probability of error to be bounded, we need,
\begin{equation}
    \lambda >\max{\left(\log(\beta), D(\mu_0 \| \hat{\mu}_0)\right)}
\end{equation}
Recall that for the JB-Page Test, we had $\lambda > 0$. However, as $n_0 \rightarrow \infty$, we can make the lower bound $\max{\left(\log(\beta), D(\mu_0 \| \hat{\mu}_0)\right)}$ and the error term $\epsilon_0$ as small as required.
\end{proof}

\subsection{Behaviour Under Post-Change Distribution}
We will now note the characteristics of this test under the post change distribution $\mu_1$.
\begin{lemma}
If $\lambda < D(\mu_1 \| \hat{\mu}_0)$, the test will terminate with probability one under $\mu_1$ (different from ${\mu}_0$). That is (recall Equation \eqref{eq:nalpha}),
\begin{equation*}
    \mu_1(N(\alpha) < \infty) = 1.
\end{equation*}
\end{lemma}
\begin{proof}
The proof for this lemma is similar to that of Lemma 3 in \cite{jacob2008} with $\mu_0$ replaced by $\hat{\mu}_0$. This replacement of $\mu_0$ with $\hat{\mu}_0$ works out in this case because the behaviour of this test under the post change distribution $\mu_1$ is identical to that of a test where the pre-change distribution was $\hat{\mu}_0$ instead of $\mu_0$.
\end{proof}

We will now state the performance of the proposed test under the distribution $\mu_1$ conditioned on the first $n_0$ samples $x_1^{n_0}$.

\begin{theorem}
For the modified CUSUM test (Test 3), and the defined stopping time, $N(\alpha)$ with $\alpha \rightarrow 0$, for $\lambda < D\left(\mu_{1} \| \hat{\mu}_0\right)$, we have,
\begin{equation}
\mathbb{E}_{\mu_{1}}(N(\alpha) \mid x_1^{n_0}) \sim|\log \alpha| /\left(D\left(\mu_{1} \| \hat{\mu}_0\right)-\lambda\right).
\end{equation}
\label{th:34}
\end{theorem}
\begin{proof}
This can be proved by arguing along similar lines as Theorem 2 in \cite{jacob2008}.
\end{proof}

\section{Performance Bounds}
To study the performance of the modified CUSUM test (Test 3), we will use the following parameter to denote the stopping time of the test.
\begin{align*}
M(\gamma)=\inf \{n: \max_{1 \leq k \leq n}&\left(-L\left(x_{n_0 + k}^{n_0 + n}\right)-\log \hat{\mu}_0\left(x_{n_0 + k}^{n_0 + n}\right) \right. \\
&\left. -n \lambda\right) \geq \log \gamma\}
\end{align*}
Note that once we have generated the empirical distribution $\hat{\mu}_0$ using the first $n_0$ samples, we can discard these samples since the test in Equation \eqref{eq:nalpha} does not take into account these symbols. Moreover, for simplicity, we now denote the symbols $x_{n_0 + 1}^{n_0 + n}$ by $y_1^{n}$. Now, the stopping time $M(\gamma)$ can be re-written as,
\begin{align*}
M(\gamma)=\inf \{n: \max_{k \leq n}&\left(-L\left(y_{k}^{n}\right)-\log \hat{\mu}_0\left(y_{k}^{n}\right) -n \lambda\right) \geq \log \gamma\}
\end{align*}
Since our test uses the first $n_0$ samples from $\mu_0$ to generate the empirical distribution $\hat{\mu}_0$, we work with the conditional expectation when studying the performance under $\mu_1$, conditioned on the initial $n_0$ long segment generated by $\mu_0$.

Now, consider the case when the first $n_0 + m$ samples come from the source $\mu_0$ and the rest of the samples come from $\mu_1$. That is, $Y_i \sim \mu_0 \forall i < m$ and  $Y_i \sim \mu_1 \forall i \geq m$. In this case, the change point for $Y_is$ is $m$. The probability measure for this variable $m$ is denoted by $P_m$. Also, $P_0$ denotes the probability measure for the case when there is no change, i.e., all the symbols are generated using the source $\mu_0$. Using a similar notation, we denote the expectation under the probability distribution $P_m$ by $E_m$. For this type of setting, Lorden introduced a minimax type criterion \cite{lorden} to measure the performance of such tests. This criterion is defined for any stopping time $N$ as we had defined earlier. The test states that,
\begin{equation*}
\bar{\mathbb{E}}_{1}(N) \triangleq \sup _{m \geq 1} ess \sup \mathbb{E}_{m}\left[(N-m+1)^{+} \mid Y_{1}, \ldots, Y_{m-1}\right].
\end{equation*}

To study the behaviour of $\bar{\mathbb{E}}_{1}(N(\gamma))$ (or, in our case, $\bar{\mathbb{E}}_{1}(N(\gamma) \mid x_1^{n_0})$), we first investigate the characteristics of $\mathbb{E}_0(M(\gamma))$. This represents the expected time after which the algorithm detects a change and stops the test, in the case when there was no actual change in the source. That is, all symbols were from the source $\mu_0$. Ideally, we want  $\mathbb{E}_0(M(\gamma))$ to be $\infty$ because the algorithm should not detect any change when there is no change.

This brings us to the following theorem, proposed in \cite{lorden}.
\begin{theorem}[Lorden]
Let N be a extended stopping variable with respect to $X_1, X_2, \ldots $, such that,
\begin{equation*}
    \mathbb{P}_0(N < \infty) \leq \alpha 
\end{equation*}
and for $k = 1, 2, \ldots$, let $N_k$ denote the stopping variable obtained by applying $N$ to $X_k, X_{k+1}, \ldots$. Then, define,
\begin{equation*}
N^{*} = \min \{N_k + k - 1 \mid k = 1, 2, \ldots\}.
\end{equation*}
Then $N^{*}$ is an extended stopping variable with,
\begin{equation*}
\mathbb{E}_0(N^{*}) \geq \frac{1}{\alpha},
\end{equation*}
and,
\begin{equation*}
\bar{\mathbb{E}}_1(N^{*}) \leq \mathbb{E}_1(N).
\end{equation*}
\label{theorem:lorden}
\end{theorem}
We also define a class of tests $S_{\gamma}$ such that if any test $S \in S_{\gamma}$, then it satisfies the property that $\mathbb{E}_0(S) > \gamma$.
Using this definition, another result on the performance of various tests of the class $S_{\gamma}$ is given as follows,
\begin{theorem}[Lorden]
For all tests $S$ in the class $\mathcal{S}_{\gamma}$, with $\gamma \rightarrow \infty$, we have,
\begin{equation}
\inf _{S \in \mathcal{S}_{\gamma}} \bar{\mathbb{E}}_{1}(S) \sim \frac{\log \gamma}{D\left(\mu_{1} \| \mu_{0}\right)}.
\end{equation}
The bound for this theorem is achieved by the CUSUM Test.
\end{theorem}
\begin{theorem}
For Test 3 in Equation \eqref{eq:nalpha}, we have,
\begin{align}
    \mathbb{E}_0(M(\gamma)) \geq \frac{\gamma}{\frac{1}{2^{\lambda - \log(\beta)} - 1} + \epsilon_0 \gamma},
    \label{eq:slto}
\end{align}
and for $\gamma \rightarrow \infty$ and $\lambda < D\left(\mu_{1} \| \hat{\mu}_0\right)$, we have,
\begin{align}
    \bar{\mathbb{E}}_{1}(M(\gamma) \mid x_1^{n_0}) \sim \frac{|\log \gamma| }{D\left(\mu_{1} \| \hat{\mu}_0\right)-\lambda}.
    \label{eq:sltt}
\end{align}
\label{th:43}
\end{theorem}
\begin{proof}
Equation \eqref{eq:slto} follows directly from Theorem \ref{theorem:lorden} and Equation \eqref{eq:pe} whereas Equation \eqref{eq:sltt} follows from Theorem \ref{theorem:lorden} and Theorem \ref{th:34}.
\end{proof}

We will now formally show that the proposed test is asymptotically optimal if the KL-Divergence between the post-change distribution and the empirical distribution is close enough to the KL-Divergence between the post-change distribution and the pre-change distribution. This is formally stated as follows.
\begin{theorem}
There exits a stopping time $M \in \mathcal{S}_{\gamma}$ such that, as $\gamma \rightarrow \infty$ we have,
\begin{equation}
\bar{\mathbb{E}}_{1}(M \mid x_1, x_2, \ldots) \sim(1+\kappa) \inf _{S \in \mathcal{S}_{\gamma}} \bar{\mathbb{E}}_{1}(S) \text{ a.s. } [\mu_0],
\end{equation}
where $\kappa$ can be chosen arbitrarily close to 0.
\end{theorem}
\begin{proof}
Let $\eta = \gamma \cdot (\frac{1}{2^{\lambda - \log(\beta)} - 1} + \epsilon_0 \gamma)$. Now, if $\epsilon_0 = \frac{1}{\gamma}$, $\eta$ can be written as $\eta = \gamma \cdot \left(\frac{1}{2^{\lambda - \log(\beta)} - 1} + 1\right)$.
Then by Equation \eqref{eq:pe} and Theorem \ref{theorem:lorden} we have,
\begin{equation*}
    \mathbb{E}_0(M(\eta)) \geq \gamma.
\end{equation*}
Now, we take $\lambda$ as (for a choice of  $\kappa>0$),
\begin{equation*}
    \lambda = D\left(\mu_{1} \| \hat{\mu}_0\right) - \frac{D\left(\mu_{1} \| \mu_0\right)}{1 + \kappa}\cdot
\end{equation*}
Also, note that as $\gamma \rightarrow \infty$, we have $\epsilon_0 \rightarrow 0$ and therefore $n_0 \rightarrow \infty$ which yields,
\begin{gather*}
    D\left(\mu_{1} \| \hat{\mu}_0\right) \xrightarrow[ ]{} D\left(\mu_{1} \| {\mu}_0\right) \text{   a.s. } [\mu_0] 
\end{gather*}
Also, as $n_0 \rightarrow \infty$, we have,
\begin{equation*}
    \bar{\mathbb{E}}_{1}(M(\eta) \mid x_1^{n_0}) \rightarrow \bar{\mathbb{E}}_{1}^{}(M(\eta) \mid x_1, x_2, \ldots) \text{ a.s. } [\mu_0].
\end{equation*}
And, by Theorem \ref{th:43}, we have,
\begin{align*}
    \bar{\mathbb{E}}_{1}(M \mid x_1^{n_0}) &\sim \frac{|\log \eta|} {D\left(\mu_{1} \| \hat{\mu}_0\right)-\lambda}
\end{align*} And thus,
\begin{align*}
    \bar{\mathbb{E}}_{1}(M \mid x_1, x_2, \ldots) &\sim (1 + \kappa) \cdot \frac{|\log \gamma|} {D\left(\mu_{1} \| \mu_0\right)} \text{ a.s. } [\mu_0]
\end{align*}
Therefore, as $n_0 \rightarrow \infty$, we have Equation (37).

Note that for the above analysis to work out (for finite $n_0$), we need $D\left(\mu_{1} \| \hat{\mu}_0\right) > \lambda > \max{\left(\log(\beta), D(\mu_0 \| \hat{\mu}_0)\right)}$, i.e.,
\begin{equation*}
    \frac{D\left(\mu_{1} \| \mu_0\right)}{1 + \kappa}  + \max{\left(\log(\beta), D(\mu_0 \| \hat{\mu}_0)\right)} < D\left(\mu_{1} \| \hat{\mu}_0\right).
\end{equation*}
Consider the case when the samples using which the empirical distribution is estimated $x_1^{n_0} \in E_0^{\delta}$, then we have,
\begin{align*}
    D(\mu_1 \| \hat{\mu}_0) - D(\mu_1 \| {\mu}_0) &= \sum_{a \in \mathcal{X}} \mu_1\left(a\right) \log \frac{\mu_0\left(a\right)}{\hat{\mu}_0\left(a\right)} \\
    &< \sum_{a \in \mathcal{X}} \mu_1\left(a\right) \log \frac{\mu_0\left(a\right)}{{\mu}_0\left(a\right) - \delta} \\
    &< \log(\beta).
\end{align*}
Therefore in this case, we have,
\begin{equation*}
    \frac{D\left(\mu_{1} \| \mu_0\right)}{1 + \kappa}  + \log(\beta) < D\left(\mu_{1} \| \hat{\mu}_0\right) < D\left(\mu_{1} \| \mu_0\right) + \log(\beta).
\end{equation*}

\end{proof}
\section{Conclusion and Future Work}
In this paper, we studied the modified JB-Page CUSUM Test for the case when the pre-change distribution is estimated using an empirical distribution and the post-change distribution is estimated using a universal compression code. We saw that as the number of samples producing $\hat{\mu}_0$ go to infinity, the test is asymptotically optimal. Currently, we are estimating the post-change distribution through a universal code and an empirical version of the pre-change distribution is used. A natural modification of this work can be to use universal compression codes to estimate the pre-change distribution as well and extend this work to markov processes.

\bibliographystyle{unsrt}
\bibliography{references}

\end{document}